\documentclass[runningheads,anonymous]{llncs}

\usepackage[utf8]{inputenc}
\usepackage{graphicx}
\usepackage{amsmath}
\usepackage{amsfonts}
\usepackage{xcolor}
\usepackage{url}
\usepackage{epic,eepic,latexsym,color}
\usepackage{ifthen,graphics,epsfig}
\usepackage{wrapfig}
\usepackage{pdfpages}
\usepackage{shuffle}
\usepackage{xspace}
\usepackage{fullpage}

\newcounter{linecounter}
\newcommand{\linenumbering}{\ifthenelse{\value{linecounter}<10}{(0\arabic{linecounter})}{(\arabic{linecounter})}}
\renewcommand{\line}[1]{\refstepcounter{linecounter}\label{#1}\linenumbering}
\newcommand{\resetline}[1]{\setcounter{linecounter}{0}#1}
\renewcommand{\thelinecounter}{\ifnum \value{linecounter} >
9\else 0\fi \arabic{linecounter}}

\begin{document}

\title{RMWable Snapshot Objects using\\Read and Write Operations for Synchronization}

\titlerunning{RMWable Snapshot Objects using Read and Write Operations for Synchronization}

\author{Armando Casta\~neda \and Braulio Ramses Hernández Martínez}

\institute{Universidad Nacional Aut\'onoma de M\'exico (UNAM)}

\maketitle              

\begin{abstract}
In the context of asynchronous concurrent shared-memory systems,
a linearizable implementation of the \emph{snapshot} object 
provides two high-level operations: $Update(k,v)$ that atomically writes $v$ to entry $MEM[k]$ of a shared array $MEM$, and $Scan()$ that atomically reads all entries of $MEM$.
Recently, the {\emph{RMWable}} snapshot object was proposed,
a variant of the snapshot object that allows processes 
to perform operations more complex than just read and write.
Specifically, each entry $MEM[k]$ is an \emph{arbitrary readable} object. 
{An interesting property of the known RMWable snapshot algorithms is that all the data used to synchronize updates and scans is \emph{decoupled} from the shared memory $MEM$, stored in a separate part of the shared memory.
Here we are interested in RMWable snapshot solutions that follow this structure.}
The known RMWable snapshot algorithms heavily rely on powerful
low-level operations such as \emph{compare\&swap} or \emph{load-link/store-conditional}
to synchronize updates and scans.
{In this work, we explore whether RMWable snapshots are possible using only \emph{read} and \emph{write} operations for synchronization, namely,  whether these simple operations are capable of synchronizing concurrent updates and scans in algorithms whose synchronization data is decoupled from $MEM$.}
We present lock-free and wait-free {RMWable snapshot algorithms},
some in the standard concurrent shared memory-model where the number of processes $n$ is finite
and known in advance, and others in a variant of the standard model with \emph{unbounded concurrency},
where there are infinitely many processes, but at any moment
only finitely many processes participate in an~execution.
{The algorithms use only $read$ and $write$ operations for synchronization of updates and scans.}
All proposed algorithms provide \emph{invisible reads}, namely, each component $MEM[k]$ can be read 
{without modifying the shared memory.}
As far as we know, our algorithms are the only RMWable snapshot implementations 
in the literature that have this feature. 
\end{abstract}

\section{Introduction}
\label{sec:intro}

Concurrent fault-tolerant shared-memory algorithms are key to exploiting the inherent parallel capabilities of modern multicore architectures~\cite{Shavit11}. In these algorithms, a collection of \emph{asynchronous failure-prone} processes communicate with one another by applying \emph{atomic} low-level operations on the entries of a shared memory, in order to collectively solve a given distributed problem. Low-level operations range from simple \emph{read} and \emph{write}, to more complex and powerful \emph{Read-Modify-Write} (RMW) operations such as \emph{fetch\&add}, \emph{compare\&swap} or \emph{load-link/store-conditional}.
As is customary in the literature on concurrent shared-memory algorithms, we adopt \emph{linearizability}~\cite{HerlihyW90} as the correctness criterion 
and focus on \emph{lock-free} and \emph{wait-free} concurrent algorithms~\cite{Herlihy91}.

\subsection*{RMWable snapshots}
Several fundamental shared objects have been identified over time, one of them being the \emph{snapshot} object~\cite{AfekADGMS93}. Generally speaking, this requires producing a copy of a shared array $MEM$ in an \emph{atomic} manner. More specifically, a snapshot implementation must provide two high-level operations: $Update(k,v)$ that atomically writes $v$ to entry $MEM[k]$, and $Scan()$ that atomically reads all entries of $MEM$. In the presence of asynchrony, the main challenge consists in designing a \emph{Scan} operation that is able to take an atomic snapshot of $MEM$ despite concurrent \emph{Update} operations executed at unpredictable times and different speeds.

The {\emph{RMWable}} snapshot object~\cite{BashariCW24,JayantiJJ24} is a recently proposed variant of the snapshot object in which processes are allowed to perform operations more complex than just read and write. In particular, each component $MEM[k]$ is an \emph{arbitrary} object: $MEM[k]$ can be a single entry of the shared memory on which atomic low-level operations can be applied (i.e., \emph{read}, \emph{write}, or an RMW operation), or a wait-free linearizable high-level object implementing a complex concurrent data structure (e.g., a priority queue or a hash table). It is assumed that each component $MEM[k]$ is \emph{readable}, that is, provides a read operation that returns its current state. While \emph{Scan} remains the same, 
$Update(k,op,args)$ now executes the operation $op(args)$ on object $MEM[k]$ and returns the result of the operation,
where $op$ is one of the operations that can be executed on $MEM[k]$ and $args$ is a valid~input~for~the~operation.

{
An interesting property of the RMWable snapshot algorithms proposed in~\cite{BashariCW24,JayantiJJ24} is that the functionality for synchronizing $Update$ and $Scan$ operations is \emph{decoupled} from the shared memory $MEM$, i.e., all data for synchronizing these operations is stored in a separate part of the shared memory.
Moreover, each component $MEM[k]$ is treated as a black-box: the type of object it implements is unknown, the only information available about it is that it is readable.
A consequence of this property is that those algorithms can be applied to any \emph{already existing} array $MEM$ with readable objects. 
Figure~\ref{fig:schematic} shows a schematic view of the structure of this type of algorithms.  
{In contrast, some snapshot implementations in the literature (e.g.~\cite{AfekADGMS93,AttiyaGR08,BashariW21,ImbsR12,WeiBBFR021})
are also in charge of \emph{implementing} the operations provided by each component $MEM[k]$, hence \emph{simulating} the contents of the whole array~$MEM$.}
As observed in~\cite{JayantiJJ24}, the {aforementioned} property makes RMWable snapshot implementations versatile, for example, allowing checkpointing existing concurrent algorithms without the need for major modifications in their code:
every operation performed on an object in $MEM$ needs to be replaced with an invocation to \emph{Update} of an RMWable implementation~with~appropriate~inputs. 

Here we are interested in RMWable snapshot algorithms that follow the structure in Figure~\ref{fig:schematic}, where $MEM$ is \emph{given} and the aim is synchronizing updates and scans with the help of a separate portion of the shared memory. 
}

\subsection*{RMWable snapshots using read and write operations}
{Loosely speaking, the algorithms in~\cite{BashariCW24,JayantiJJ24} use \emph{compare\&swap} or \emph{load-link/store-conditional} to keep, at all times, a consistent snapshot states of the objects in $MEM$, stored in a separate shared array $MEM'$; while the size of $MEM'$ is the same as $MEM$ in~\cite{JayantiJJ24}, the size of $MEM'$ is larger than the size in $MEM$ in~\cite{BashariCW24}. In this way, \emph{Scan} somehow can efficiently retrieve a snapshot of $MEM$ from $MEM'$. The correctness of the algorithms heavily relies on the high coordination power of \emph{compare\&swap} and \emph{load-link/store-conditional} to keep $MEM'$ in a consistent state.} These operations are known to have \emph{infinite} coordination power, which means that they allow any number of processes to coordinate and solve any distributed problem~\cite{Herlihy91}.

\begin{figure}[t]
\begin{center}
\includegraphics[scale=0.2]{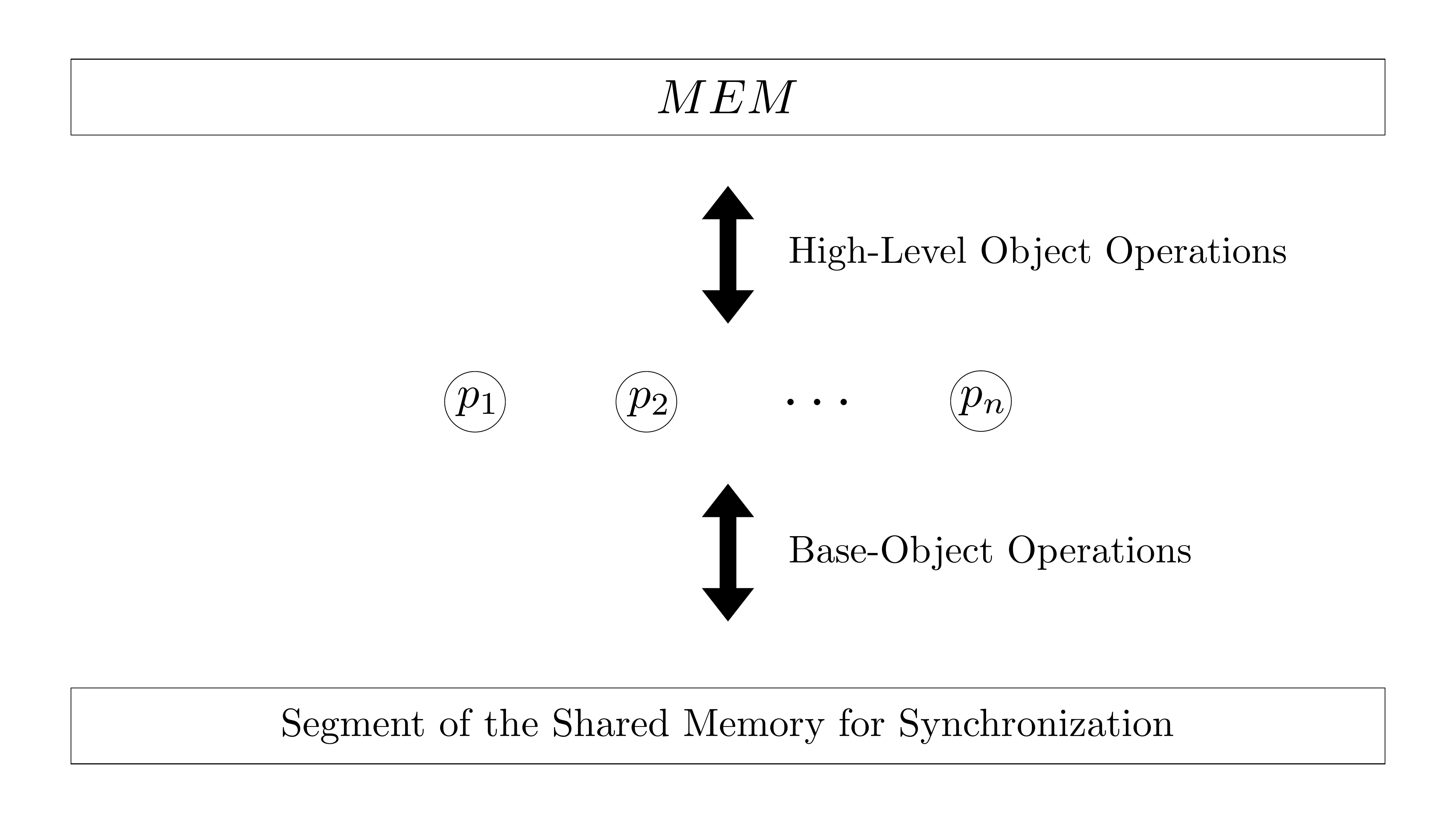}
\caption{A schematic view of an RMWable snapshot implementation whose synchronization data is decoupled from $MEM$.}
\label{fig:schematic}
\end{center}
\end{figure}

Are powerful low-level operations needed to snapshot collections of arbitrary readable objects, namely, {for synchronizing the $Update$ and $Scan$ operations of algorithms that follow the structure in Figure~\ref{fig:schematic}?} 
Are the simplest low-level operations, \emph{read} and \emph{write}, capable of providing the required synchronization for implementing this type of snapshots?
Recall that \emph{read} and \emph{write} operations are at the bottom of the \emph{consensus hierarchy}~\cite{Herlihy91}, hence it is impossible to solve consensus among two or more processes
using only these~operations.

{In this work, we are interested in RMWable snapshot algorithms whose functionality and data for synchronizing updates and scans, decoupled from $MEM$ itself, is based only on \emph{read} and \emph{write} operations.
 Of course, some objects in $MEM$ can internally use RMW low-level operations, as it is known that implementing some high-level objects, {for example, queues and stacks, requires the use of RMW low-level operations~\cite{Herlihy91}}.
However, we do not consider those operations as part of the RMWable snapshot algorithms.} 

\subsection*{Results and their implications}
First, we present two {RMWable snapshot algorithms} in the standard asynchronous shared memory model, where the number of processes $n$ is finite and known in advance, as in~\cite{BashariCW24,JayantiJJ24}. 
{The algorithms use only $read$ and $write$ operations for synchronization of updates and scans.}
One of the algorithms is lock-free and the other is a wait-free version of the lock-free one enriched with a helping mechanism. 
{These two algorithms are based on an adaptation 
of the \emph{double collect} technique used in the seminal snapshot algorithm in~\cite{AfekADGMS93} to our decoupled setting.}
Then, different from~\cite{BashariCW24,JayantiJJ24}, we consider a variant of the model with \emph{unbounded concurrency}~\cite{MerrittT13}, where there are infinitely many processes but at any moment only finitely many processes participate in an execution. 
Our interest in unbounded concurrency is that it naturally requires algorithms to adapt to the number of processes currently participating, preventing them from relying on the total number of processes in the system, and guaranteeing progress even if processes keep joining an execution.
{With the help of the \emph{collect} algorithm in~the~unbounded~concurrency~model in~\cite{GafniMT01},
we extend our wait-free algorithm in the standard model to the unbounded concurrency model.} 

{A drawback of the known RMWable snapshot algorithms~\cite{BashariCW24,JayantiJJ24} is that they do not provide \emph{invisible reads}~\cite{InvisibleReads}.
{An implementation provides invisible reads if it is possible to read any specific component $MEM[k]$ without modifying 
the shared memory.}
Providing invisible reads is a desirable property of an RMWable snapshot algorithm: one might expect only light synchronization when reading a specific component of $MEM$. 
Moreover, as noted in~\cite{InvisibleReads}, such functionality is desirable in practical settings, 
as it is often the case that most operations on shared memory are reads.\footnote{Invisible reads reduce cache-coherence
contention in practical implementations, since $read$ operations do not access cache lines containing shared variables in an exclusive mode~\cite{architecture}.}}

A recent impossibility result~\cite{InvisibleReads} shows that 
the algorithms in~\cite{BashariCW24,JayantiJJ24} possess properties that actually \emph{prevent} them from providing invisible reads:
{it is impossible to add a $Read(k)$ operation to these algorithms that provides invisible reads.
In contrast, our algorithms provide a $Read(k)$ operation that only performs $MEM[k].read()$. 
Hence, whenever $MEM[k].read()$ does not modify the shared memory, our algorithms provide an invisible read,
and the read is optimal in the sense that only $MEM[k].read()$ is performed.}

{{We observe that} 
our algorithms in the standard model imply the following in message-passing systems 
where less than half of the processes can crash. 
Let $MEM$ be a collection of linearizable, readable objects in such a system, say $S$. 
As already mentioned, implementing some objects requires consensus to be possible.
It is known that consensus is impossible in the presence of asynchrony and crash failures~\cite{FischerLP85}.
Hence the system $S$ might be synchronous or semi-synchronous.
Our algorithms, which use only \emph{read} and \emph{write} operations for synchronization, 
and the simulation in~\cite{AttiyaBD95}, which simulates \emph{read/write} shared-memory on top of an 
asynchronous message-passing system in which less than half of the processes can crash, 
imply that none of the timing assumptions of $S$ are needed to produce snapshots of $MEM$ {in~$S$
(of course, internally each component of $MEM$ might rely on the~timing~assumptions~of~$S$).}}

We stress that standard techniques used in previous snapshot algorithms do not apply, 
at least directly, in our context. For example, the double collect technique in~\cite{AfekADGMS93} 
cannot be directly used since object states in $MEM$ cannot be tagged with \emph{timestamps} or {\emph{handshake bits}}. 
The approaches in~\cite{BashariCW24,JayantiJJ24}, briefly described above, 
are impossible due to the weak coordination~power~of~\emph{read/write}~operations.

Finally, our main focus is computability. Our algorithms are not intended to achieve time or space optimality.
Clarity and simplicity in the solutions are our main target. Motivated by the checkpointing scenario mentioned above,
however,  the final part of the paper discusses a variant of the wait-free algorithm 
in the standard model that, {with the help of \emph{fetch\&add} base-objects}, 
improves the step complexity of $Update$ when there are no concurrent $Scan$ operations. 

\subsection*{More related work}
Snapshot algorithms can be partitioned into two classes~\cite{AfekADGMS93}: \emph{multi-writer}, where each process can update every $MEM[k]$, and \emph{single-writer}, where the size of $MEM$ is equal to the number of processes and each process is assigned to a unique entry of $MEM$, which is the only entry it can update. The algorithms presented here fall into the more general multi-writer class.

Numerous single-writer or multi-writer snapshot algorithms using $read$ and $write$ operations have been proposed in concurrent models with a finite number of processes (e.g.,~\cite{AfekADGMS93,AttiyaR98,ImbsR12}). 
{All such algorithms that we are aware of simulate the memory $MEM$.} 
Basically, each $MEM[k]$, in addition to data, stores synchronization data that allows the snapshot algorithm to consistently perform scans of the data encoded in $MEM$. A typical technique in these algorithms is that of storing \emph{timestamps} or {\emph{handshake bits}} as synchronization data, which are helpful for detecting if an entry of $MEM$ has changed. 
{Generally, these algorithms assume shared variables that are large enough to store data and some of the synchronization data (e.g., a process ID, a timestamp, a handshake bit, a snapshot, {or a pointer to a snapshot}).
{We are not aware of a transformation that decouples data and synchronization data in these algorithms and preserves correctness.} 
There are snapshot implementations that use small shared variables, some of them using shared variables of size $\Theta(\log n)$, where $n$ denotes the number of processes~\cite{small}.} 
{The algorithms presented here do not require large shared variables, as explained in Section~\ref{sec:final}.}

There are snapshot algorithms that avoid this issue by using 
more powerful low-level operations such as \emph{compare\&swap} or \emph{load-link/store-conditional}
(e.g.,~\cite{AttiyaGR08,BashariW21,FatourouK07,RianyST01}).
Most of them simulate $MEM$ too,
but more efficiently, as they exploit the coordination power of those low-level operations. 
The snapshot algorithm in~\cite{Jayanti05}, based on \emph{load-link/store-conditional}, is time- and space-optimal,
and its synchronization functionality is decoupled from the memory $MEM$, in the style of the RMWable snapshot algorithms in~\cite{BashariCW24,JayantiJJ24}.

Two algorithms provide alternative forms of RMWable snapshot~\cite{Jayanti02,WeiBBFR021}. The snapshot algorithm in~\cite{WeiBBFR021} uses \emph{compare\&swap} low-level operations to simulate a memory $MEM$ whose entries provide \emph{read}, \emph{write} and \emph{compare\&swap} high-level operations. 
As far as we know, the first snapshot algorithm whose functionality for synchronizing updates and scans is decoupled from $MEM$ appears in~\cite{Jayanti02}. This algorithm, based on \emph{load-link/store-conditional} low-level operations, allows updates to apply RMW operations on the entries of $MEM$ but does not provide scans, instead it atomically applies a function $f$ whose input is the sequence of values of the components of $MEM$ and returns the result. {If $f$ is the identity function, the algorithm indeed returns a snapshot.}

For the sake of efficiency, some snapshot or RMWable snapshot algorithms provide \emph{partial} snapshots~\cite{AttiyaGR08}, namely, {they are able to
provide the values of a specified sequence of components of $MEM$, rather than all components.} Other snapshot algorithms are \emph{adaptive}~\cite{BashariCW24,BashariW21,JayantiJJ24,WeiBBFR021}, 
i.e., they split \emph{Scan} into two operations, $Click$ that returns a reference to a snapshot, and $Observe(j)$ that retrieves the state in position $j$ of the most recent snapshot of the process (taken with $Click$). 
Since we focus on computability, we do not study these variants.

Therefore, the difficulties we face in this paper are: {(1) no synchronization data of an RMWable snapshot algorithm can be stored together with object states in $MEM$, as $MEM$ is given and each component is a black-box}, and (2) \emph{read} and \emph{write} operations provide limited capabilities for process coordination. 

It was recently shown~\cite{InvisibleReads} that, under some conditions, invisible reads are impossible in adaptive RMWable snapshot algorithms. 
Invisible reads are impossible whenever 
(1) the algorithm is \emph{obstruction-free}~\cite{obstruction-free}, 
a progress property that is strictly weaker than lock-freedom and wait-freedom, 
(2) $Click$ accesses fewer than $\min\{n-2, m\}$ objects of $MEM$,
where $n$ is the number of processes and $m$ is the size of $MEM$,
(3) the linearization point of an $Update$ is independent of $Click$ and any other $Update$, and
(4) objects in $MEM$ {are decoupled from the portion of the shared memory for synchronization and component $MEM[k]$ can only be modified through $Update(k,op,args)$.}
{As shown in~\cite{InvisibleReads},} the algorithms in~\cite{BashariCW24,JayantiJJ24} satisfy the four properties, 
hence {it is impossible to add a $Read(k)$ operation to them that only performs $MEM[k].read()$ and reads from the shared memory used for syncrhnization.}
Our algorithms do not have $Click$ and $Observe$ operations (as our focus is not efficiency),
but they can be easily implemented through $Scan$: 
$Click$ simply invokes $Scan$ and stores the result in a persistent local variable of the process and $Observe$ responds using the variable.
The resulting algorithms in the model with a finite number of processes satisfy all properties but~(2), as $Click$ accesses all objects in
$MEM$, thus they are not subject to the impossibility result in~\cite{InvisibleReads}.
{Therefore, our results show that it is possible to evade the impossibility result by dropping only~one~of~the~four~properties.}

It has been shown that it is possible to implement the snapshot object 
in models with unbounded concurrency, using only \emph{read} and \emph{write} operations~\cite{GafniMT01}. 
Prior to this work, the RMWable snapshot object was not studied in this context. 
Thus, our algorithms expand the computability boundary in the unbounded concurrency model.
The reader is referred to~\cite{Aguilera04} for a survey on the research in models with infinitely many processes.

%\subsection*{Organization of the paper}
%The rest of the paper is organized as follows. 
%Section~\ref{sec:prelim} introduces the model of computation, as well as the notions of linearizability, lock-freedom and wait-freedom, and the definition of the RMWable snapshot object. Section~\ref{sec:single} presents a restricted solution that guarantees correctness as long as there are no concurrent updates. This simple algorithm serves to explain the main ideas on which the general solution in Section~\ref{sec:multi} builds on. Then, Section~\ref{sec:unbounded} presents the second RMWable algorithm, which works in the model with unbounded concurrency. Section~\ref{sec:final} concludes the paper with a final discussion, which includes a variant of the algorithm in Section~\ref{sec:multi} that improves step complexity, and possible future directions.

\section{Preliminaries}\label{sec:prelim}

\subsection*{Model of computation}
We consider a standard concurrent system with \emph{n asynchronous} processes, denoted $p_1,...,p_n$, which may \emph{crash} at any time during an execution of the system (see for example~\cite{AfekADGMS93,BashariCW24,Herlihy91,JayantiJJ24}). It is assumed that all processes but one can crash in an execution. The \emph{index} of process~$p_i$~is~$i$. Processes communicate with one another by invoking \emph{atomic} operations on shared \emph{base objects}. A base object can provide atomic \emph{read/write} operations (such an object is henceforth called a \emph{register}), or more powerful atomic \emph{Read-Modify-Write} (RMW) operations, such as \emph{fetch\&add} and \emph{compare\&swap}.

A {\emph{concurrent implementation}, or simply \emph{implementation}}, of a (high-level) {sequentially specified} object $T$ (e.g., a queue or a stack) is a distributed algorithm $\mathcal{A}$ consisting of {state machines} $A_1,...A_n$. {Machine} $A_i$ specifies which operations on base objects $p_i$ executes in order to return a response when it invokes a high-level operation of $T$. Each of these base object operations is a \emph{step}.

A \emph{configuration} $C$ of the system contains the states of
all shared base objects and processes.
In the \emph{initial} configuration,
base objects and processes are in their initial states.
An \emph{execution} of $\mathcal{A}$ is a possibly infinite sequence
$C_0 e_1 C_1 e_2 C_2 \cdots$, where $C_0$ is the initial configuration, 
each $e_i$ is a step of a process or an invocation or a response of a high-level operation by a process,
and $C_i$ is the configuration that results of executing $e_i$ when the configuration of the system is $C_{i-1}$.

In an execution of $\mathcal{A}$, an operation is \emph{complete} if both its invocation and response appear in the execution; an operation is \emph{pending} if only its invocation appears in the execution. A process is \emph{correct} in an infinite execution if it takes infinitely many steps in the execution.
{We assume that whenever a process $p_i$ completes an operation, it non-deterministically picks a new operation to be executed.}

We say that $\mathcal{A}$ is \emph{wait-free}~\cite{Herlihy91} if in every infinite execution, every correct process completes all its operations (hence every correct process completes infinitely many operations). We say that $\mathcal{A}$ is \emph{lock-free}~\cite{Herlihy91} if in every infinite execution, infinitely many operations are complete (hence at least one correct process completes infinitely many operations). 

\emph{Linearizability} is the standard correctness condition for concurrent implementations of objects defined via sequential specifications. The reader is referred to~\cite{HerlihyW90} for the detailed definition. For our purposes here, the following simplified definition based on \emph{linearization points} suffices.
We say that a distributed algorithm $\mathcal{A}$ is a \emph{linearizable} implementation of a sequential object $T$ if each of its finite executions satisfies the following: for all complete operations and some pending operations, it is possible to append responses to those pending operations and find unique points, one per operation and each point placed between the invocation and response of the associated operation, such that the linearization points induce a sequential execution that is {\emph{valid} for the sequential object $T$, i.e., it satisfies the sequential specification of object $T$.}

While Sections~\ref{sec:single} and~\ref{sec:multi} consider the model just stated with finite $n$, Section~\ref{sec:unbounded} considers the \emph{unbounded concurrency}~\cite{MerrittT13} version of the model where $n=\infty$. An execution in this model is again a possibly infinite sequence of steps. Thus, any finite prefix of an infinite execution has a finite number of \emph{participating} processes, namely, the processes that have taken at least one step in the prefix.

\subsection*{RMWable snapshots}
The classic (multi-writer) \emph{snapshot} object~\cite{AfekADGMS93} allows processes to {write} to the entries of a shared array, and also {read} the whole array. More specifically, the state of the snapshot sequential object is an array $MEM[1,..,m]$, which is modified through two operations: $Update(k,v)$, that places $v$ in $MEM[k]$, where $1 \leq k\leq m$ and $v$ is taken from the space of values that can be written in the entries of $MEM$; and $Scan()$, that reads all values in $MEM$ and return the values.

We are interested in the {\emph{RMWable snapshot}} object~\cite{BashariCW24,JayantiJJ24}, a generalization of the snapshot object described above. In this generalization, each component $MEM[k]$ stores an arbitrary atomic base-object or a wait-free linearizable object implementation.
Each component $MEM[k]$ is assumed to be \emph{readable}, namely, it provides a \emph{read} operation that returns the current state of $MEM[k]$, with no side effects.\footnote{Formally, any read operation can be removed from the linearization of $MEM[k]$ (recall that each $MEM[k]$ is assumed to be linearizable), and the resulting sequential history is a valid sequential history of the object $MEM[k]$ implements.} 
{As explained in Section~\ref{sec:intro}, 
we are interested in RMWable snapshot implementations that follow the structure that appears in Figure~\ref{fig:schematic}, i.e., 
the functionality for synchronizing $Update$ and $Scan$ operations is decoupled from $MEM$.}

More in detail, we consider that the RMWable snapshot sequential object provides two operations: {$Update(k,op,args)$}, that executes operation $op(args)$ on $MEM[k]$ and returns the result of the operation (without loss of generality, we assume that every operation returns a value), 
and {$Scan()$}, that reads the states of all objects in $MEM$ and returns the states.
{Different from~\cite{BashariCW24,JayantiJJ24}, we consider a third operation, {$Read(k)$}, that performs $MEM[k].read()$ and returns the result.}

{Following~\cite{InvisibleReads}, a $Read$ operation of an RMWable snapshot implementation is \emph{invisible}
{if it does not modify the shared memory.}
Whenever $MEM[k].read()$ does not modify the shared memory, our algorithms will provide invisible reads;
{furthermore, the invisible reads will 
only perform $MEM[k].read()$}, without even reading from the synchronization portion of the shared memory.}

\section{A simple solo-updater solution}
\label{sec:single}

We first present a simple solution for the \emph{solo-updater} case, where it is assumed that $Update$ operations 
\emph{do not execute concurrently}. 
This solution, which appears in~Figure~\ref{fig:solo-updater}, exposes the main ideas 
the general \emph{concurrent-updater} algorithm in the next section is based on. 
First, we consider the version of the algorithm without shared variable $H$ and without lines \ref{A3}, \ref{A7} and~\ref{A9}. 
This algorithm is just lock-free, and once its correctness is settled, it will be clear that the removed lines make the algorithm wait-free. 
The algorithm~in~Figure~\ref{fig:solo-updater}, as well as algorithms coming after, uses the \emph{non-atomic} function $Collect(M)$
that, given a finite shared array $M$ of readable objects, asynchronously reads the entries of~$M$, one by one in some arbitrary order, 
and returns an array with the read values.

 %==================================================================
\begin{figure}[t]%algorithm}[htb]
\centering{\fbox{
\begin{minipage}[t]{150mm}
\scriptsize
\renewcommand{\baselinestretch}{2.5} \resetline
\begin{tabbing}
aaa\=aa\=aa\=aa\=aa\=aa\=aa\=\kill %~\\

{\bf Shared Variables:}\\

$~~$  {$MEM[1,.., m]:$ array of linearizable, wait-free, readable objects}\\ 

$~~$  $T:$ register, initialized to $0$\\

$~~$  {\color{red} $H:$ register, initialized to $\bot$}\\ \\

{\bf Operation} {$Update(k, op, args)$:}\\

\line{A1} \> \> $t  \leftarrow {T.read()}$\\

\line{A2} \> \> ${T.write(t+1)}$ \hspace{0.2cm} {\color{gray}\%\% Beginning of operation}\\

\line{A3} \> \> {\color{red} $H.write(Collect(MEM))$}  \hspace{0.2cm} {\color{gray}\%\% For helping scans}\\

\line{A4} \> \> $res \leftarrow MEM[k].op(args)$\\

\line{A5} \> \> ${T.write(t+2)}$ \hspace{0.2cm} {\color{gray}\%\% End of operation}\\

\line{A6} \> \> {\bf return} $res$ \\

{\bf End} $Update$\\ \\

{\bf Operation} {$Scan()$:}\\

\line{A7} \> \> {\color{red} $t_0 \leftarrow {T.read()}$} \hspace{0.2cm} {\color{gray}\%\% Current state of updates}\\

\line{A8} \> \> {\bf while $\sf true$ do} \\

\line{A9} \> \> \> {\color{red} {\bf if} $({T.read()} - t_0) \geq 3$ {\bf then return} {$H.read()$} } {\color{gray}\%\% Help can be adopted}\\

\line{A10} \> \> \> $t \leftarrow {T.read()}$ \hspace{0.2cm} {\color{gray}\%\% State of updates before collect}\\

\line{A11} \> \> \> $s \leftarrow Collect(MEM)$\\

\line{A12} \> \> \> $t' \leftarrow {T.read()}$ \hspace{0.2cm} {\color{gray}\%\% State of updates after collect}\\

\line{A13} \> \> \> {\bf if} $t' - t \leq 1$ {\bf then return} $s$ {\bf end if} \hspace{0.2cm} {\color{gray}\%\% $s$ is a snapshot}\\

\line{A15} \> \> {\bf end while} \\

{\bf End} $Scan$\\ \\

{\bf Operation} {$Read(k)$:}\\

\line{A16} \> \> {\bf return} $MEM[k].read()$ \\

{\bf End} $Read$

\end{tabbing}
\end{minipage}
  }
\caption{A solo-updater algorithm. {Algorithm} for process $p_i$.
Local variables of the process are denoted in lowercase letters.
While the algorithm without the lines marked in red is only lock-free, 
the algorithm with those lines is wait-free.}
\label{fig:solo-updater}
}
\end{figure}%algorithm}
%=================================================================

In the lock-free algorithm, the processes share a variable $T$ that ``measures'' the progress of the process that
is currently modifying $MEM$. In $Update$, the process announces the beginning of its operation by incrementing $T$, 
then modifies $MEM[k]$, afterwards it announces the end of its operation by incrementing $T$ once again, 
and finally returns the result of the operation on $MEM[k]$. 
In $Scan$, the process performs a sequence of collects of $MEM$, 
each collect in between two reads from $T$,
until it detects in line~\ref{A13} that its current collect is a snapshot.
$Read$ simply executes $MEM[k].read()$ and returns the result.

Proving linearizability of the lock-free algorithm is simple. 
{The condition in line~\ref{A13} aims to find an interval of time $I$ where $MEM$ is modified \emph{at most once}.
When the condition in that line is true, the interval $I$ goes from the read in line~\ref{A10} to the read of line~\ref{A12} in the current iteration of the loop.}
Since $T$ is initialized to~$0$ and it is assumed that there are no concurrent $Update$ operations, 
it follows that:

\begin{itemize}
\item if $t$ and $t'$ are equal to the same even value, then it is certain that no process modifies $MEM$ in interval $I$;

\item in all other cases, namely, $t$ and $t'$ are equal to the same odd value, or their difference is one, 
$MEM$ is modified at most once in $I$.
\end{itemize}

The previous simple observation implies that the collect obtained in line~\ref{A11} can be linearized: 

\begin{itemize}
\item if $MEM$ is not modified, the $Scan$ operation can be linearized at 
any read step of the collect;
\item otherwise it is linearized at any step of the collect before or after the modification of $MEM$, depending on whether the collect
reads $MEM$'s modified component before or after the modification (see Figure~\ref{fig:snapshot}).
\end{itemize}

\begin{figure}[t]
\begin{center}
\includegraphics[scale=0.23]{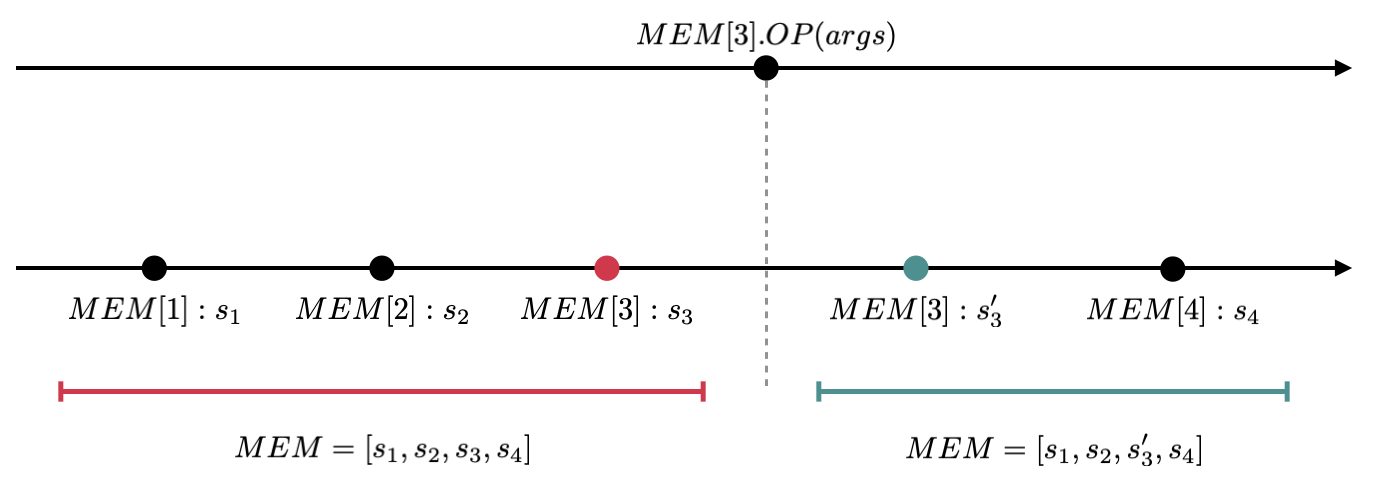}
\caption{If $MEM$ is modified at most once during a collect, the collect is indeed a snapshot. 
In the example, $MEM$ has four objects and while the collect is taken,
a process performs an operation on $M[3]$, changing the state of the object from $s_3$ to~$s'_3$. 
Step $MEM[k].read()$ returning $s$ is simply denoted $MEM[k]:s$. 
If the collect reads $M[3]$ before $M[3]$ is modified, then the values obtained by the collect, $[s_1,s_2,s_3,s_4]$, coincide with the state 
of the objects of $MEM$ during the interval that goes from the first read of the collect to the read of $M[3]$;
if the collect reads $M[3]$ after the modification, the collect obtains the states in $MEM$, $[s_1,s_2,s'_3,s_4]$, 
during the interval that goes from
the read of $M[3]$ to the last read of the collect.}
\label{fig:snapshot}
\end{center}
\end{figure}

\begin{theorem}
\label{theo:solo-updater-lf}
The algorithm~in~Figure~\ref{fig:solo-updater} without shared variable $H$ and without lines \ref{A3}, \ref{A7} and \ref{A9}
is a {linearizable, lock-free, solo-updater implementation of the RMWable snapshot object that
uses only read and write low-level operations for high-level~operations~synchronization.}
\end{theorem}

\begin{proof}
{Observe that only read and write operations are used for synchronizing updates, scans and reads.} 
It is clear that $Update$ and $Read$ are wait-free. 
It is also clear that the only reason a $Scan$ operation of a non-crashed process never finishes, is because 
there are infinitely many $Update$ operations that complete. Hence the algorithm is lock-free. 

Consider any finite execution of the algorithm. 
To prove linearizability, we linearize every $Update$ operation that modifies $MEM$ at its step where it modifies it, i.e., line~\ref{A4},
and if the operation is pending, the corresponding response is appended;
any pending $Update$ that does not modify $MEM$ is discarded. 
Similarly, every $Read$ operation that reads $MEM$ is linearized at its step that reads $MEM[k]$, i.e., line~\ref{A16},
and any pending $Read$ operation is discarded.
$Scan$ operations that return are linearized as described below.
First, all pending $Scan$ operations are discarded.
Recall that, as observed above, the condition in line~\ref{A13} guarantees that $MEM$ is modified
at most one time while the collect $C$ in line~\ref{A11} that obtains the array $s$ returned by the $Scan$ is taken. We have two cases:
\begin{itemize}
\item $MEM$ is not modified while $C$ is taken. In this case, clearly $s$ is a snapshot of $MEM$,
and the operation can be linearized in a step of the collect.

\item $MEM$ is modified while $C$ is taken. As already mentioned, in this case $MEM$ is modified exactly once. 
Let $MEM[k]$ be the component that is modified while the collect $C$ is taken. The operation is linearized at the step of
$C$ that reads $MEM[k]$. Observe that if $C$ reads $MEM[k]$ before (resp. after) $MEM[k]$ is modified, 
then $s$ contains the states of the objects in $MEM$ before (resp. after) the modification (see Figure~\ref{fig:snapshot}).
\end{itemize}
By construction, the chosen linearization points induce a valid sequential execution of the RMWable sequential object.
Moreover, the linearization point of every operation is placed between the invocation and response of the operation.
Therefore, every execution of the algorithm is linearizable.
\qed
\end{proof}

We now consider the full algorithm~in~Figure~\ref{fig:solo-updater}. 
Shared variable $H$ and lines \ref{A3}, \ref{A7} and \ref{A9} implement a \emph{helping} mechanism that makes the algorithm wait-free.
In the style of several snapshot algorithms (e.g.,~\cite{AfekADGMS93,ImbsR12}), an $Update$ operation helps $Scan$ operations
by performing a scan itself, whose result is stored in~$H$. Due to the assumption that there are no concurrent
$Update$ operations, any collect performed in line~\ref{A3} is a snapshot.
A process $p_i$ executing a $Scan$ records in line~\ref{A7} the value of $T$ at the beginning of its operation, from its perspective, 
and line \ref{A9} indicates to $p_i$ that it can use the snapshot stored in $H$ only if
$T$ has been incremented at least three times. 
Since every complete $Update$ operation increments~$T$ two times, 
the condition in line \ref{A9} guarantees that the snapshot stored in~$H$ that $p_i$ returns is produced 
by an $Update$ nested in $p_i$'s current $Scan$ operation.
This property is important for obtaining linearizability.
Observe that, due to asynchrony,  
the snapshot that $p_i$ returns might be different from the one stored in $H$ when $p_i$ verifies that 
$T - t_0 \geq 3$ in line~\ref{A9}. This is not a problem because that snapshot is produced by an $Update$ nested in $p_i$'s operation.

\begin{theorem}
\label{theo:solo-updater-wf}
The algorithm~in~Figure~\ref{fig:solo-updater} is a 
{linearizable, wait-free, solo-updater implementation of the RMWable snapshot object that
uses only read and write low-level operations for high-level~operations~synchronization.}
\end{theorem}

\begin{proof}
{Observe that only read and write operations are used for synchronizing updates, scans and reads.} 
Due to the conditions in lines~\ref{A9} and~\ref{A13},
every $Scan$ operation executes the loop at most three times. Hence the algorithm is wait-free.

We now prove linearizability. $Update$ operations and $Scan$ operations that return in line~\ref{A13} 
are linearized as in the proof of Theorem~\ref{theo:solo-updater-lf}.
$Scan$ operations that return in line~\ref{A9} are linearized as follows.
Consider any such operation $op$ that returns array $s'$ (read from $H$). 
We argue that there is a collect taken in line~\ref{A3} and returning $s'$
that is nested in $op$. If we prove that, then $op$ can be linearized at any step of that collect,
from which linearizability follows, as every collect stored in $H$ is a snapshot (due to the solo-updater assumption).
The argument is that, since the condition in line~\ref{A9} is true, $T$ is incremented at least three times
from the beginning of $op$ until $p_i$ reads $s'$ from $H$. 
As every complete $Update$ operation increments $T$ two times, 
it follows that the snapshot read by~$p_i$ in line~\ref{A9} is taken by an $Update$ operation
that is nested in $op$. Thus, the linearization point of $op$ lies between its invocation and response.
Therefore, the algorithm is linearizable.
\qed
\end{proof}

\section{The general case}\label{sec:multi}
We now consider the concurrent-updater case, where concurrent $Update$ operations are allowed.
Figure~\ref{fig:concurrent-updater} presents a wait-free, linearizable RMWable snapshot implementation. 
As in the previous section, we first consider a short version of the solution.
We show that the algorithm without shared array~$H$ 
and without lines~\ref{B3},~\ref{B7} and~\ref{B9} is linearizable but only lock-free.

In the lock-free algorithm, each process $p_i$ has a shared variable $T[i]$ that measures the progress of
its $Update$ operations. Basically, $Update$ operates per-process as in the solo-updater lock-free algorithm in the previous section.
Once again, $Read$ simply reads $MEM[k]$ and returns the result.
In $Scan$, the scanning process performs a sequence of collects of $MEM$, 
each collect preceded and followed by a collect of~$T$,
 until the condition in line~\ref{B13} is satisfied.
The discussion in the previous section shows that when this condition is satisfied,
\emph{every} process modifies $MEM$ at most once in the interval of time that goes from the last read of the collect in line~\ref{B10}
to the first read of the collect in line~\ref{B12},~in~the~current~iteration~of~the~loop. 

As explained in the previous section, 
if it is the case that $t[j]$ and $t'[j]$ are equal to the same even value, 
then $p_j$ does not modify $MEM$ when the collect $C$ in line~\ref{B11}, returning~$s$, is performed. 
Thus, line~\ref{B14} computes the maximum number~$\ell$ of processes that might have modified $MEM$ during~$C$, each of them at most once.
If $\ell$ is at most one, line~\ref{B15}, then $s$ is indeed a snapshot, as argued in the analysis of the solo-updater algorithm.
Otherwise, it might be possible that~$s$ is not a snapshot (see Figure~\ref{fig:no-snapshot}). 
In such a case, in lines~\ref{B17}-\ref{B18'}, the scanning process 
takes at most~$\lfloor \ell/2 \rfloor$ additional collects of $MEM$, 
requiring each of them to be equal to $s$.
If so, the process takes a third collect of~$T$, on line~\ref{B19}, and if this collect is equal to the second collect of $T$
(meaning that no process started a new $Update$ operation in the meantime),
then $s$ is indeed a snapshot and is returned on line~\ref{B20}.
The reason is that $MEM$ is modified at most~$\ell$ times while 
 the consecutive and equal $\lfloor \ell/2 \rfloor+1$ collects of $MEM$ are taken, hence, by the pigeonhole principle,
there must be at least one in which $MEM$ is modified at most once.
Thus, the $Scan$ operation can be linearized at one of the steps of that collect, 
as already~argued. 

We believe that the mechanism just described can be understood as the equivalent of a double collect~\cite{AfekADGMS93}
in the context of  RMWable snapshot algorithms that follow the structure in Figure~\ref{fig:schematic}
and use only read and write operations for synchronization of updates, scans~and~reads.

\begin{theorem}
\label{theo:concurrent-updater-lf}
The algorithm~in~Figure~\ref{fig:concurrent-updater} without shared array $H$ and without lines~\ref{B3},~\ref{B7} and~\ref{B9}
is a {linearizable, lock-free implementation of the RMWable snapshot object that
uses only read and write low-level operations for high-level~operations~synchronization.}
\end{theorem}

\begin{proof}
{Observe that only read and write operations are used for synchronizing updates, scans and reads.} 
 It is clear that $Update$  and $Read$ are wait-free. 
Observe that the only reason a $Scan$ operation of a non-crashed process never finishes, is because 
there are infinitely many $Update$ operations that complete. Thus, the algorithm~is~lock-free. 

Consider any finite execution of the algorithm.
To prove linearizability, we linearize $Update$ and $Read$ operations, and $Scan$ operations that return in line~\ref{B15}
as in the linearizability proof of the solo-updater lock-free algorithm.
That is to say, every $Update$ that modifies $MEM$ is linearized at its step where it modifies it, i.e., line~\ref{B4},
and every complete $Read$ is linearized at its step that reads $MEM$, i.e., line~\ref{B24}.
As for $Scan$ operations, as argued in the analysis of the solo-updater lock-free algorithm,
in each loop iteration, at most $\ell$ processes ($\ell$ computed in line~\ref{B14}) 
modify $MEM$ during the interval of time
that lapses from the last read of the collect in line~\ref{B10} to the first read of the collect in line~\ref{B12}.
Thus, if the condition in line~\ref{B15} is satisfied, then at most one process modifies $MEM$ 
while the collect $C$ in line~\ref{B11}
that produces~$s$ is taken, implying that $s$ is indeed a snapshot. 
Thus, the $Scan$ operation can be linearized at a step of $C$.

We now deal with the $Scan$ operations that return in line~\ref{B20}. 
Let $op$ be any such operation. First, observe that 
since arrays $t'$ and $t''$ are equal, line~\ref{B20}, 
then $\ell$ still upper bounds the number of processes that modify $MEM$
between the first and third collect of $T$.
More precisely,  at most $\ell$ processes modify $MEM$ during the interval of time
that lapses from the last read of the collect in line~\ref{B10} to the first read of the collect in line~\ref{B19}.
Therefore, by the pigeonhole principle, 
among the consecutive and equal $\lfloor \ell/2 \rfloor+1$ collects taken in that interval, lines~\ref{B11} and~\ref{B18},
there is at least one, say $C$, in which $MEM$ is modified at most once. As observed above, this implies that $s$ indeed
a snapshot of $MEM$ and $op$ can be linearized at a step of $C$.

By construction, the chosen linearization points induce a valid sequential execution of the RMWable sequential object.
Moreover, the linearization point of every operation is placed between the invocation and response of the operation.
Therefore, every execution of the algorithm is linearizable.
\qed
\end{proof}

%==================================================================
\begin{figure}[t]%algorithm}[htb]
\centering{\fbox{
\begin{minipage}[t]{150mm}
\scriptsize
\renewcommand{\baselinestretch}{2.5} \resetline
\begin{tabbing}
aaa\=aa\=aa\=aa\=aa\=aa\=aa\=\kill %~\\

{\bf Shared Variables:}\\

$~~$  {$MEM[1,.., m]:$ array of linearizable, wait-free, readable objects}\\ 

$~~$  $T[1,.., n]:$ array of registers, initialized to $[0,..,0]$\\ 

$~~$  {\color{red} $H[1,.., n]:$ array of registers, initialized to $[\bot,..,\bot]$}\\ \\

{\bf Operation} ${Update(k, op, args)}$:\\

\line{B1} \> \> $t \leftarrow {T[i].read()}$\\

\line{B2} \> \> ${T[i].write(t+1)}$ \hspace{0.2cm} {\color{gray}\%\% Beginning of operation}\\

\line{B3} \> \> {\color{red} $H[i].write(Scan(MEM))$} \hspace{0.2cm} {\color{gray}\%\% For helping scans}\\

\line{B4} \> \> $res \leftarrow MEM[k].op(args)$\\

\line{B5} \> \> ${T[i].write(t+2)}$ \hspace{0.2cm} {\color{gray}\%\% End of operation}\\

\line{B6} \> \> {\bf return} $res$ \\

{\bf End} $Update$\\ \\

{\bf Operation} ${Scan()}$:\\

\line{B7} \> \> {\color{red} $t_0 \leftarrow Collect(T)$} \hspace{0.2cm} {\color{gray}\%\% Current state of updates}\\

\line{B8} \> \> {\bf while $\sf true$ do} \\

\line{B9} \> \> \> {\color{red} {\bf if} $\exists j, ({T[j].read()} - t_0[j]) \geq 3$ {\bf then return} ${H[j].read()}$ {\bf end if}} {\color{gray}\%\% Help can be adopted}\\

\line{B10} \> \> \> $t \leftarrow Collect(T)$ \hspace{0.2cm} {\color{gray}\%\% State of updates before collect}\\

\line{B11} \> \> \> $s \leftarrow Collect(MEM)$\\

\line{B12} \> \> \> $t' \leftarrow Collect(T)$ \hspace{0.2cm} {\color{gray}\%\% State of updates after collect}\\

\line{B13} \> \> \> {\bf if} $\forall j, t'[j] - t[j] \leq 1$ {\bf then} \hspace{0.2cm} {\color{gray}\%\% Every process modified $MEM$ at most once}\\

\line{B14} \> \> \> \> $\ell \leftarrow n - |\{j : t[j] = t'[j] \, \wedge \, t[j] \hbox{is even}\}|$\\

\line{B15} \> \> \> \> {\bf if} $\ell \leq 1$ {\bf then return} $s$ \hspace{0.2cm} {\color{gray}\%\% At most one process modified $MEM$}\\

\line{B16} \> \> \> \> {\bf else}\\

\line{B17} \> \> \> \> \> {\bf for} $\lfloor \ell/2 \rfloor$ {\bf times do}  \hspace{0.2cm} {\color{gray}\%\% Number of collects adapts to $\ell$}\\

\line{B18} \> \> \> \> \> \> {\bf if} $s \neq Collect(MEM)$ {\bf then go to} line~\ref{B8} {\bf end if}\\

\line{B18'} \> \> \> \> \> {\bf end for}  \\

\line{B19} \> \> \> \> \> $t'' \leftarrow Collect(T)$ \hspace{0.2cm} {\%\% State of updates after sequence of collects}\\

\line{B20} \> \> \> \> \> {\bf if} $t' = t''$ {\bf then return} $s$ {\bf end if} \hspace{0.01cm} {\color{gray}\%\% There exists a snapshot}\\

\line{B21'} \> \> \> \> {\bf end ifelse}\\

\line{B21''} \> \> \> {\bf end if}\\

\line{B23} \> \> {\bf end while} \\

{\bf End} $Scan$\\ \\

{\bf Operation} ${Read(k)}$:\\

\line{B24} \> \> {\bf return} $MEM[k].read()$ \\

{\bf End} $Read$

\end{tabbing}
\end{minipage}
  }
\caption{A concurrent-updater algorithm. {Algorithm} for process $p_i$.
Local variables are denoted in lowercase letters.
While the algorithm without the lines marked in red is only lock-free, the algorithm with those lines is wait-free. 
}
\label{fig:concurrent-updater}
}
\end{figure}%algorithm}
%=================================================================

We observe that the algorithm that results from removing lines~\ref{B16}-\ref{B21'} from the
lock-free algorithm {
(i.e., the algorithm~in~Figure~\ref{fig:concurrent-updater} without shared array $H$ and without lines~\ref{B3},~\ref{B7} and~\ref{B9})} 
is linearizable but \emph{blocking}, namely, 
it has infinite executions where, from some time on, no operation completes. 
{As an example, consider an execution in which $2 \leq x \leq n-1$ processes concurrently start
$Update$ operations and crash after executing line~\ref{B2}, and then a process $q$ starts a $Scan$ operation.
{Observe that $T$ does not change from that point on.}
The operation of $q$ never terminates as in every iteration of the loop the value computed in line~\ref{B14} is $x \geq 2$.}
Thus, effectively, lines~\ref{B16}-\ref{B21'} are the mechanism that provides lock-freedom.
{Intuitively, this blocking algorithm can be understood as
$n$ copies of the lock-free solo-updater solution 
(i.e., the algorithm~in~Figure~\ref{fig:solo-updater} without shared variable~$H$ and without lines \ref{A3}, \ref{A7} and~\ref{A9}),
where each updater has its own shared variable~$T$.}

Now, shared array $H$ together with lines~\ref{B3},~\ref{B7} and~\ref{B9} implement a helping mechanism
that makes the algorithm~in~Figure~\ref{fig:concurrent-updater} wait-free. Helping is basically the one in the solo-updater algorithm 
in the previous section. In $Update$, the updating process $p_i$ performs a scan, this time through $Scan$, which is stored in $H[i]$, line~\ref{B3}.
In $Scan$, for each process $p_j$, the scanning process $p_i$ reads the value in $T[j]$ at the beginning of its operation, from its perspective, line~\ref{B7}; 
if at some time $T[j]$ has been incremented at least three times, $p_i$ borrows the snapshot of $p_j$ in $H[j]$, line~\ref{B9}.
As observed in the analysis of the solo-updater algorithm, the condition in line~\ref{B9} guarantees that
the snapshot stored in $H[j]$ is taken by an $Update$ operation (of~$p_j$) that is nested in the current $p_i$'s $Scan$ operation,
which~provides~linearizability.

\begin{figure}[t]
\begin{center}
\includegraphics[scale=0.23]{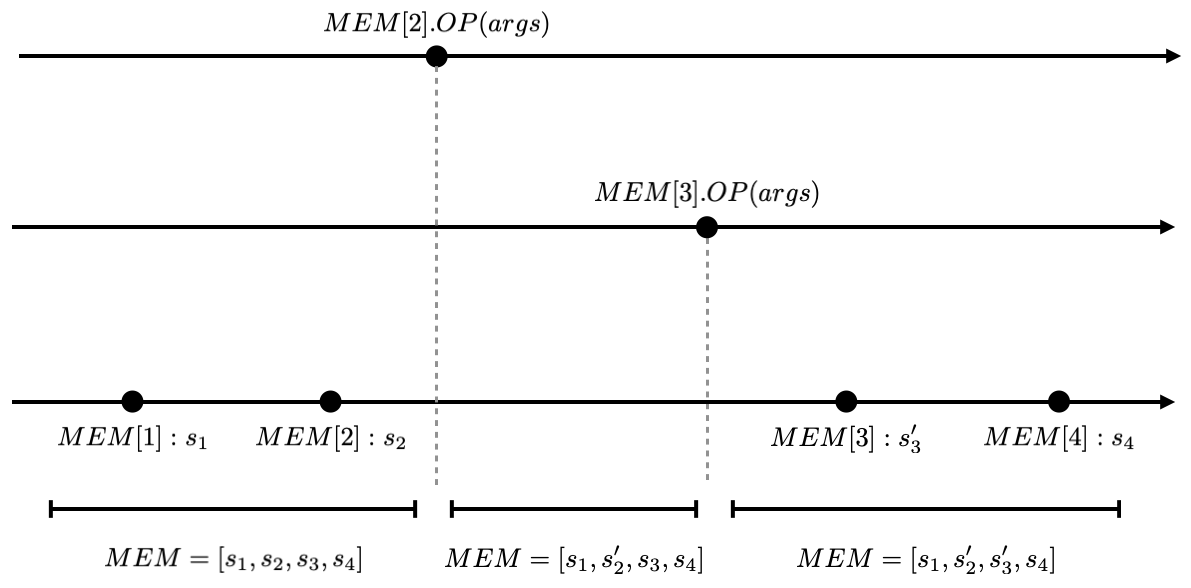}
\caption{If $MEM$ is modified two or more times during a collect, the collect might not be a snapshot. 
In the example, $MEM$ has four objects and while the collect is taken,
 objects $MEM[2]$ and $MEM[3]$ are modified,
changing states from $s_2$ and $s_3$ to $s'_2$ and $s'_3$, respectively. 
Step $MEM[k].read()$ returning $s$ is simply denoted $MEM[k]:s$. 
Observe that the values obtained by the collect, $[s_1,s_2,s'_3,s_4]$,
do not correspond to the object states in $MEM$ at any point in time, hence it is not a snapshot.}
\label{fig:no-snapshot}
\end{center}
\end{figure}

\begin{theorem}
\label{theo:concurrent-updater-wf}
The algorithm~in~Figure~\ref{fig:concurrent-updater} is a {linearizable, wait-free implementation of the RMWable snapshot object that
uses only read and write low-level operations for high-level~operations~synchronization.}
\end{theorem}

\begin{proof}
{Observe that only read and write operations are used for synchronizing updates, scans and reads.}  
Observe that if a $Scan$ operation of a process $p_i$ does not return in lines~\ref{B15} or~\ref{B20}
in one of its loop iterations, then a process, say $p_j$, either increments $T[j]$ or modifies $MEM$, in one of its $Update$ operations during the loop.
Hence, the same $Update$ operation of $p_j$ prevents $p_i$ from returning in lines~\ref{B15} or~\ref{B20} for at most three loop iterations of $p_i$'s operation, and consequently after at most three $Update$ operations
of $p_j$, $p_i$ is guaranteed seeing that the condition in line~\ref{B9}
holds for $j$. 
Therefore, after $O(n)$ iterations, $p_i$ is guaranteed seeing that the condition in line~\ref{B9} holds for some $j$, hence returning. 

Consider any finite execution of the algorithm.
We now prove linearizability. $Update$ and $Read$ operations, and $Scan$ operations that return in lines~\ref{B15} or~\ref{B20} 
are linearized as in the proof of Theorem~\ref{theo:concurrent-updater-lf}.
$Scan$ operations that return in line~\ref{A9} are treated as follows.
Consider any such operation $op$ that returns an array $s'$, read from $H[j]$. 
First, we argue that there is an $Update$ operation $op'$ of $p_j$ 
that is nested in $op$ and whose $Scan$ operation in line~\ref{B3} returns $s'$.
Since the condition in line~\ref{B9} is true, $T[j]$ is incremented at least three times
from the beginning of $op$ until $p_i$ reads $s'$ from $H[j]$. 
As every complete $Update$ operation increments its corresponding entry of $T$ two times, 
it follows that the snapshot read by $p_i$ in line~\ref{B9} is taken 
by the $Scan$ in an $Update$ operation $op'$ of $p_j$ that is nested in~$op$. 

We now argue that the linearization point of the snapshot produced by the $Scan$ in operation $op'$
lies between the invocation and response of $op'$.
We do this by induction on the depth $d$ of the help. 
Depth is $d=0$ if the $Scan$ in $op'$ takes its snapshot on its own (i.e. returns in line~\ref{B15} or~\ref{B20}),
depth is $d=1$ if the $Scan$ in $op'$ borrows a snapshot (i.e. returns in line~\ref{B9}) taken by 
the $Scan$ in an $Update$ that takes its snapshot on its own (i.e. with depth 0).
In general, depth is $d$ if the $Scan$ in $op'$ borrows a snapshot taken by the $Scan$ in an $Update$ 
with depth $d-1$.

For $d = 0$ the claim is obvious as the snapshot is taken by the $Scan$ operation on its own.
Consider now depth $d > 0$. By definition, the $Scan$ in $op'$ borrows a snapshot taken by the 
$Scan$ in an $Update$ operation $op''$ with depth $d-1$. By induction hypothesis,
the linearization point of the snapshot taken in $op''$ lies between the invocation and response of $op''$.
We have already shown that $op''$ is nested in $op'$. Hence the induction step follows.

Finally, the operation $op$ is linearized at the linearization point of the $Scan$ operation executed by $op'$.
This completes the linearization proof, because the proof of Theorem~\ref{theo:concurrent-updater-lf} already showed
that the linearization points of $Update$ and $Read$ operations, and $Scan$ operations that return in lines~\ref{B15} or~\ref{B20} 
induce a valid sequential execution of the RMWable snapshot object.
Therefore, the implementation is linearizable.
\qed
\end{proof}

 \section{A solution in the unbounded concurrency model}\label{sec:unbounded}
 
Next we consider a model with unbounded concurrency~\cite{MerrittT13}, 
where there are infinitely many processes, $p_1, p_2, \hdots$,
and in every infinite execution, only a finite number of processes participate in any prefix.
It has been shown that collects of infinite arrays are possible in this model of computation~\cite{GafniMT01}.
The algorithm~in~Figure~\ref{fig:unbounded-concurrency} adapts the wait-free concurrent-updater algorithm in the previous
section to the context of the unbounded concurrency model. The algorithm makes the following assumptions:

\begin{itemize}
\item The first step of a process $p_i$ in an execution writes $0$ to $T[i]$,
which announces the process is participating in the execution. 
{Note that setting $T[i]$ to 0 does not affect any scan, it only indicates that $p_i$ is in the execution but has not started any update yet.}

\item There is a function $Collect(T)$ that implements a 
collect of infinite array $T$~\cite{GafniMT01}. 
We also can assume that the function returns a set~$t$ with a pair $(j,v)$
for every participating process $p_j$, where $v \geq 0$ is the value the collect reads from $T[j]$.
We let $t.ids$ denote the set of the process indexes
appearing in $t$, namely, the set $\{j : (j, -) \in t \}$, 
and for every $j \in t.ids$, we let $t[j]$ denote the value $v \geq 0$ such that $(j,v) \in t$.
\end{itemize}

Intuitively, the main difficulty in adapting the algorithm~in~Figure~\ref{fig:concurrent-updater} to the unbounded concurrency model is that,
in every iteration of a $Scan$ operation, new processes might join the execution, start $Update$ operations and modify $MEM$
in a way so that none of the collects in lines~\ref{B11} and~\ref{B18} {(of Figure~\ref{fig:concurrent-updater})}
is a snapshot, and then just crash.
In this way the $Scan$ operation never sees the condition in line~\ref{B9} {(of Figure~\ref{fig:concurrent-updater})} true, hence running forever.

To deal with this issue, {in Figure~\ref{fig:unbounded-concurrency}}, $Scan$ will treat 
processes that join the execution in the course of the operation differently from
those that are in the execution before the operation starts. The helping mechanism is enhanced with the
following simple observation: if a process $p_j$ joins the execution and starts an $Update$  
after a $Scan$ operation of $p_i$ starts,
then $p_i$ can return $H[j]$ at any time $H[j] \neq \bot$, regardless of the value of $T[j]$.
Enhancing the helping mechanism in this way preserves the wait-freedom of $Scan$,
as in the scenario described above, the processes joining the execution
and preventing a $Scan$ from taking a snapshot by itself, first perform a scan of $MEM$ before modifying it,
hence unblocking the $Scan$ operations that they~might~have~affected.

In Figure~\ref{fig:unbounded-concurrency}, while $Update$ and $Read$ remain the same,
$Scan$ now operates as follows (we just explain the main differences with $Scan$ in~Figure~\ref{fig:concurrent-updater}).
The scanning process $p_i$ keeps in a local set $init$ the set of process indexes in the first collect of $T$ that it takes on {line~\ref{C7}}.
From the point of view of $p_i$, this is the set of processes that are in the execution when its current operation starts.
Process $p_i$ stores in $curr$ the indexes of the processes that are currently in the execution, which is kept up to date at the end
of every iteration of the main loop, line~\ref{C31}.
The helping mechanism is enriched in line~\ref{C10}. 
The snapshot in $H[j]$ of a process $p_j$ can be adopted if $p_j$ increments $T[j]$
at least three times (as in~Figure~\ref{fig:concurrent-updater}); 
the second case when $H[j]$ can be adopted is if $p_j$ starts 
after the $Scan$ operation starts (namely, $j \notin init$) and there is already a snapshot in~$H[j]$.

\newpage

 %==================================================================
\begin{figure}[hb]%algorithm}[htb]
\centering{\fbox{
\begin{minipage}[t]{150mm}
\scriptsize
\renewcommand{\baselinestretch}{2.5} \resetline
\begin{tabbing}
aaa\=aa\=aa\=aa\=aa\=aa\=aa\=\kill %~\\

{\bf Shared Variables:}\\

$~~$  {$MEM[1,.., m]:$ array of linearizable, wait-free, readable objects}\\ 

$~~$  $T[1,2,..]:$ array of registers, initialized to $[-1,-1,..]$\\ 

$~~$  $H[1,2,..]:$ array of registers, initialized to $[\bot,\bot,..]$\\ \\

{\bf When process $p_i$ joins the execution:}\\

\>  \>  ${T[i].write(0)}$ \\ \\

{\bf Operation} ${Update(k, op, args)}$:\\

\line{C1} \> \> $t \leftarrow {T[i].read()}$ \\

\line{C2} \> \> ${T[i].(t+1)}$ \hspace{0.2cm} {\color{gray}\%\% Beginning of operation}\\

\line{C3} \> \> ${H[i].write(Scan(MEM))}$ \hspace{0.2cm} {\color{gray}\%\% For helping scans}\\

\line{C4} \> \> $res \leftarrow MEM[k].op(args)$\\

\line{C5} \> \> ${T[i].write(t+2)}$ \hspace{0.2cm} {\color{gray}\%\% End of operation}\\

\line{C6} \> \> {\bf return} $res$ \\

{\bf End} $Update$\\ \\

{\bf Operation} ${Scan(MEM)}$:\\

\line{C7} \> \> ${t_0, init \leftarrow collect(T)}$ {\color{gray}\%\% Current state of updates}\\

\line{C8} \> \> $curr \leftarrow \emptyset$\\

\line{C9} \> \> {\bf while $\sf true$ do} \\

\line{C10} \> \> \> {\bf if} $\exists j \in curr, ({T[j].read()} - t_0[j]) \geq 3 \, \vee \, (j \notin init \, \wedge \, {H[j].read()} \neq \bot)$ {\bf then}\\

\line{C10'} \> \> \>  \> {\bf return} {$H[j].read()$} \hspace{0.2cm} {\color{gray}\%\% Help can be adopted}\\

\line{C10''} \> \> \> {\bf end if}\\

\line{C11} \> \> \> $t \leftarrow Collect(T)$ \hspace{0.2cm} {\color{gray}\%\% State of updates before collect}\\

\line{C12} \> \> \> $s \leftarrow Collect(MEM)$\\

\line{C13} \> \> \> $t' \leftarrow Collect(T)$ \hspace{0.2cm} {\color{gray}\%\% State of updates after collect}\\

\line{C15} \> \> \> {\bf if} $(\forall j \in t.ids :  t'[j] - t[j] \leq 1) \, \wedge \, (\forall j \in t'.ids \setminus t.ids : t'[j] \leq 2)$ {\bf then} \\

\line{C15'} \> \> \> \> {$a \leftarrow \{j \in t'.ids \setminus t.ids : {H[j].read()} = \bot \}$}\\

\line{C16} \> \> \> \> $\ell \leftarrow |t'.ids| - |\{j \in t.ids :  t[j] = t'[j] \, \wedge \, t[j] \hbox{is even} \}| - {|a|}$\\

\line{C17} \> \> \> \> {\bf if} $\ell \leq 1$ {\bf then return} $s$ {\color{gray}\%\% At most one participating process modified $MEM$}\\

\line{C18} \> \> \> \> {\bf else}\\

\line{C19} \> \> \> \> \> {\bf for} $\lfloor \ell/2 \rfloor$ {\bf times do}  \hspace{0.2cm} {\color{gray}\%\% Number of collects adapts to $\ell$}\\

\line{C20} \> \> \> \> \> \> {\bf if} $s \neq Collect(MEM)$ {\bf then go to} line~\ref{C9} {\bf end if}\\

\line{C21} \> \> \> \> \> {\bf end for} \\

\line{C22} \> \> \> \> \> $t'' \leftarrow Collect(T)$ {\color{gray}\%\% State of updates after sequence of collects}\\

\line{C23} \> \> \> \> \> {\bf if} $t' = t'' \, \wedge \, {(\forall j \in a, {H[j].read()} = \bot)} $  {\bf then return} $s$ {\bf end if}  \hspace{0.01cm} {\color{gray}\%\% $\exists$ a snapshot}\\

\line{C24} \> \> \> \> \> $t' \leftarrow t''$\\

\line{C25} \> \> \> \> {\bf end ifelse}\\

\line{C26} \> \> \> {\bf end if} \\

\line{C31} \> \> \> $curr \leftarrow t'.ids$\\

\line{C32} \> \> {\bf end while} \\

{\bf End} $Scan$\\ \\

{\bf Operation} {$Read(k)$}:\\

\line{C33} \> \> {\bf return} $MEM[k].read()$ \\

{\bf End} $Read$
\end{tabbing}
\end{minipage}
  }
\caption{A wait-free RMWable snapshot algorithm for the unbounded concurrency model. {Algorithm} 
for process $p_i$. Local variables of the process are denoted in lowercase letters.}
\label{fig:unbounded-concurrency}
}
\end{figure}%algorithm}
%=================================================================

Due to the new processes that might join an execution, in an iteration of the loop, 
process indexes $t.ids$ and $t'.ids$ might be distinct, 
however, it is always the case that $t.ids \subseteq t'.ids$. It similarly occurs with $t''.ids$.
Recall that the aim of the condition in line~\ref{B13} in~Figure~\ref{fig:concurrent-updater} is to detect
that every process modifies $MEM$ at most once while the collect of $MEM$ in line~\ref{B11} is taken.
The condition in line~\ref{C15} in~Figure~\ref{fig:unbounded-concurrency} implements the same idea, 
taking into account the new processes joining the execution between
the collects in lines~\ref{C11} and~\ref{C13}, namely, for every $j \in t'.ids \setminus t.ids$, 
$t'[j]$ has to be no more than two.
Line~\ref{C16} computes the bound $\ell$ on the number of processes that modify $MEM$ at most once while
the collect $C$ in line~\ref{C12} is performed. 
Since in $Update$ a process first performs a scan and then modifies $MEM$,
for any new process $p_j$ (i.e., with $j \in t'.ids \setminus t.ids$), it is certain that  if $H[j] = \bot$,
$p_j$ did not modify $MEM$ while $C$ is performed.
The rest of the operation is basically the same. {If $\ell$ is no more than one, 
then $s$, performed in line~\ref{C12}, is a snapshot, hence is returned, on line~\ref{C17}.
Otherwise, the process performs a sequence of $\lfloor \ell/2 \rfloor$ collects, requiring each of them to be equal to~$s$, on lines~\ref{C19}-\ref{C21}.
If true, the process performs one more collect of $T$, on line~\ref{C22}, and if all processes are in the same state, namely,
the sets $t'$ and $t''$ are equal and any process in set $a$ has not taken a snapshot yet, 
then $s$ is a snapshot and is returned on line~\ref{C23}.}

\begin{theorem}
\label{theo:unbounded-concurrency-wf}
The algorithm~in~Figure~\ref{fig:unbounded-concurrency} is a  {linearizable, wait-free implementation of the RMWable snapshot object 
in the unbounded concurrency model of computation that uses only read and write low-level operations for high-level~operations~synchronization.}
\end{theorem}

\begin{proof}
The linearizability proof is basically the same as the one for the algorithm~in~Figure~\ref{fig:concurrent-updater} given above. 
The only extra observation is that if a $Scan$ operation $op$ returns $H[j]$ in line~\ref{C10}, for a 
$j \notin init$, then \emph{all} $Update$ operations of $p_j$ so far are nested in $op$, regardless of the current value of $T[j]$.

Thus, we just need to argue that $Scan$ is wait-free, which proves that $Update$ remains wait-free.
By contradiction, suppose that there is a $Scan$ 
operation of a correct process $p_i$ that never terminates in an execution. This is not possible if from
some time on no new process joins the execution, as in this case the behavior of the algorithm is basically the
same as that of the algorithm~in~Figure~\ref{fig:concurrent-updater}, which we already know is wait-free.
Thus, it must be that infinitely many processes join the computation,
and modify $MEM$ while the infinitely many collects in lines~\ref{C12} and~\ref{C20} are taken,
in a way that condition in lines~\ref{C10} and~\ref{C17} are never true.
But in $Update$, the updating process first performs a scan, storing the result in~$H$, and then modifies $MEM$.
Therefore, eventually $p_i$ observes that there is a $j \notin init$ such that $H[j] \neq \bot$,
hence terminating its $Scan$ operation in line~\ref{C10}. A contradiction.
This completes the proof of the theorem.
\qed
\end{proof}

\section{Final discussion}\label{sec:final}

A typical property of interest of concurrent shared memory algorithms is that of using small shared registers. 
This is the case of the algorithm~in~Figure~\ref{fig:concurrent-updater}.\footnote{Here we do not
consider the algorithm~in~Figure~\ref{fig:unbounded-concurrency}, since most of snapshot 
and RMWable snapshot algorithms are designed for the standard model, 
and our aim is to compare our work to those algorithms.}
Each $T[i]$ is just a counter, hence if it is assumed that it can store $b$ bits, for some reasonable $b > 0$,
it lets $p_i$ execute $O(2^b)$ operations before it overflows.\footnote{In any real-world system, 
setting $b$ to 64, or even 32, would be enough for any practical matter.}
Each $H[i]$ can be a pointer to the current helping snapshot of~$p_i$
represented as a data structure (e.g., a linked list).

We have studied {RMWable snapshots} from a computability perspective, 
leaving efficiency on the side. 
The step complexity of the algorithm~in~Figure~\ref{fig:concurrent-updater} is polynomial in $n$ and $m$, 
the number of processes and the size of $MEM$, respectively. 
The step complexity of $Scan$ is $O(n^2 m)$:
each loop iteration takes $O(nm)$ steps {(counting any operation on $MEM[k]$ as a single step)} and, as shown in the proof of Theorem~\ref{theo:concurrent-updater-wf}, there are $O(n)$ iterations.
Hence, $Update$ has the same step complexity.
{In contrast, in~\cite{BashariCW24,JayantiJJ24}, $Update$ has constant step complexity in both algorithms,
and a snapshot  can be taken (through $Click$ and $Observe$) in $O(m \log n)$ time in~\cite{BashariCW24} and $O(m)$ time in~\cite{JayantiJJ24}.
 As for space, the algorithm~in~Figure~\ref{fig:concurrent-updater} uses $O(nm)$ shared variables {in addition to the space used in $MEM$}.
This is better than the space complexity of the RMWable snapshot algorithm in~\cite{BashariCW24},
whose space complexity is $O(nm \log n)$. The space complexity of the algorithm in~\cite{JayantiJJ24}
is $O(m)$, but it allows only a designated process to invoke~$Scan$.}

Polynomial step complexity is considered high in the context of snapshot algorithms.
Particularly, this high step complexity is unreasonable in the context of checkpointing already existing concurrent algorithms
mentioned in Section~\ref{sec:intro}. 
In these situations, we might expect that $Update$ and $Read$ are 
invoked frequently, and $Scan$ is invoked occasionally. 
Thus, we would like to have an RMWable snapshot algorithm whose $Update$ operation
has low step complexity (ideally constant) and $Read$ operation provides invisible reads 
(ideally optimal invisible reads).
Arguably, we might be willing to tolerate a $Scan$ operation with high step complexity,
as it is only invoked from time to time.
Appendix~\ref{sec:fast-updates} shows a variant of the algorithm in~Figure~\ref{fig:concurrent-updater} that,
with the help of two $fetch\&add$ base-objects, provides an $Update$ operation whose step complexity
is constant in executions with no concurrent~$Scan$~operations.

{As explained in Section~\ref{sec:intro}, the algorithms in Figures~\ref{fig:concurrent-updater} and~\ref{fig:concurrent-updater-improved} (the second one in Appendix~\ref{sec:fast-updates})
show that it is possible to evade the impossibility result in~\cite{InvisibleReads}
by dropping only one of the four properties identified in the result.}
Specifically, they drop the property requiring accessing fewer than $\min\{n-2,m\}$ objects of $MEM$
in $Click$ (implemented through $Scan$).\footnote{In adaptive snapshots, the aim of splitting $Scan$ 
is to design fast $Click$ and $Observe$ operations, e.g., of complexities $o(m)$ and $O(\log m)$, respectively.
Thus, in our approach for evading the result in~\cite{InvisibleReads}, splitting $Scan$ into two operations is pointless.} 
It would be interesting to explore 
if invisible reads are possible if one of the other properties is dropped. 
Arguably, any reasonable RMWable snapshot algorithm should be obstruction-free
and should modify $MEM$ only through $Update$ operations. 
Thus, the only reasonable option that arguably remains
to be explored is if invisible reads are possible when the remaining property is dropped,
i.e., linearization points of updates are independent of clicks~and~other~updates.
{If impossible, the result would say that any obstruction-free and adaptive RMWable snapshot implementation
that provides invisible reads and modifies $MEM$ only through $Update$ operations (i.e., it follows the structure in Figure~\ref{fig:schematic}),
must have a $Click$ operation that accesses a large number of objects of $MEM$, which~would~make~its~property~of~being~adaptive~pointless.}

\subsection*{Acknowledgements}
We would like to thank Hagit Attiya for comments on an early version of this work,
as well as to the anonymous reviewers, and specially Faith Ellen for several 
insightful comments and suggestions that helped improve the final version of the paper.

\bibliographystyle{plain}
\bibliography{references}

\appendix

\section{A solution with fast updates through \emph{fetch\&add}}
\label{sec:fast-updates}

The algorithm~in~Figure~\ref{fig:concurrent-updater} can be easily modified to provide an $Update$
operation that runs fast when there are no concurrent $Scan$ operations.
The modification appears in~Figure~\ref{fig:concurrent-updater-improved}.
The idea of the algorithm is that $Update$ operations trigger the (slow) helping mechanism only if 
there are concurrent $Scan$ operations requesting help. 
This is implemented through two $fetch\&add$ base-objects, $HELP$ and $TIME$:
while $HELP$ serves for the processes to announce that they are performing $Scan$, 
$TIME$ is a global clock that serves for the processes to timestamp helping snapshots.
$Scan$ first records the current time and then announces
that it is starting its operation, lines~\ref{D10} and~\ref{D11}, respectively. 
In this way, if $Update$ sees in line~\ref{D3} that $HELP$ was
greater than zero in line~\ref{D2'}, then it activates the helping mechanism, lines~\ref{D4} and~\ref{D5}:
after incrementing $TIME$, it performs a scan and stores the result in $H[i]$, together with a timestamp.\footnote{In a practical
implementation, a process might suspend for a certain amount of time before setting $help$ to ${\sf true }$ in line~\ref{D2'}, in order to
let active $Scan$ operations  possibly terminate on their own, and then test again after resuming if there are still $Scan$ operations
in progress; if so, it then sets $help$ to ${\sf true}$.}
Hence, $Scan$ can adopt the snapshot in $H[j]$ 
if the timestamp of the snapshot is larger than the initial time read by the scanning process, lines~\ref{D13} to~\ref{D16}.
Before returning, $Scan$ announces that it no longer requires help, lines~\ref{D32}~and~\ref{D33}.

\begin{theorem}
\label{theo:concurrent-updater-wf-improved}
The algorithm~in~Figure~\ref{fig:concurrent-updater-improved} is a linearizable, 
wait-free implementation of the RMWable snapshot object, whose $Update$ and $Scan$ have step complexity $O(n^2m)$.
Moreover, if an $Update$ runs with no concurrent $Scan$, 
it takes only $\Theta(1)$ steps.
\end{theorem}

\begin{proof}
For linearizability, first observe that if a $Scan$ operation $op$ executes line~\ref{D14}, 
then $s$ is produced by an $Update$ that is concurrent to $op$, and moreover
its $Scan$ that produces $s$ is nested in $op$. The reason is that $Update$
first increments $TIME$ and then invokes $Scan$, lines \ref{D4} and~\ref{D5}.
After this observation, the very same reasoning in the linearizability proof of Theorem~\ref{theo:concurrent-updater-wf}
proves that the algorithm~in~Figure~\ref{fig:concurrent-updater-improved} is linearizable. 

We now prove wait-freedom. 
We first consider $Scan$. By contradiction, suppose that there is an execution where 
a $Scan$ operation $op$ of a correct process never terminates. 
Then, eventually $HELP > 0$ forever. The only reason $op$ never terminates is because there
are infinitely many $Update$ operations that complete and terminate. But some of these operations necessarily 
produce snapshots whose timestamps are larger than $t_0$ of $op$. 
Hence, eventually $op$ terminates, which is a contradiction. 
Therefore, we have that $Scan$ is wait-free. This implies that $Update$ is wait-free too.
$Read$ is obviously wait-free.

The same step complexity analysis for the algorithm~in~Figure~\ref{fig:concurrent-updater} in Section~\ref{sec:final}
gives that the step complexity of $Scan$ and $Update$ is $O(n^2 m)$.
Finally, if $Update$ is not concurrent to a $Scan$ operation, then it reads $0$ from $HELP$ in line~\ref{D2'},
and hence does not call $Scan$. This concludes the proof.
\qed
\end{proof}

We observe that if only a designated process can invoke $Scan$, as in~\cite{JayantiJJ24},
the algorithm~in~Figure~\ref{fig:concurrent-updater-improved} can be turned \emph{read/write} for synchronization of updates and scans as follows.
Shared variable $HELP$ is replaced with a register, initialized to $0$, and $Scan$ sets it to $1$ in line~\ref{D11}
and to $0$ in line~\ref{D32}. Shared variable $TIME$ is replaced with an $n$-entry array of registers, 
each initialized to $0$. Line~\ref{D4} stores $\max(Collect(TIME))+1$ in $time$,
and line~\ref{D10} stores $\max(Collect(TIME))$ in $t_0$. The rest of the algorithm remains the same. 
The resulting algorithm satisfies the properties stated in Theorem~\ref{theo:concurrent-updater-wf-improved}.

%\newpage

%==================================================================
\begin{figure}[hb]%algorithm}[htb]
\centering{\fbox{
\begin{minipage}[t]{150mm}
\scriptsize
\renewcommand{\baselinestretch}{2.5} \resetline
\begin{tabbing}
aaa\=aa\=aa\=aa\=aa\=aa\=aa\=\kill %~\\

{\bf Shared Variables:}\\

$~~$  {$MEM[1,.., m]:$ array of linearizable, wait-free, readable objects}\\ 

$~~$  $T[1,.., n]:$ array of registers, initialized to $[0,..,0]$\\ 

$~~$  $H[1,.., n]:$ array of registers, initialized to $[\bot,..,\bot]$\\ 

$~~$  $HELP, TIME:$ $fetch\&add$ base-objects, initialized to $0$\\ \\

{\bf Operation} ${Update(k, op, args)}$:\\

\line{D1} \> \> $t \leftarrow {T[i].read()}$ \\

\line{D2'} \> \> {\bf if} ${HELP.read()} > 0$ {\bf then} $help \leftarrow {\sf true}$ {\bf end if}\\

\line{D2} \> \> ${T[i].write(t+1)}$ \\

\line{D3} \> \> {\bf if} $help$ {\bf then}\\

\line{D4} \> \>  \> $time \leftarrow TIME.fetch\&add(1)$\\

\line{D5} \> \>  \>  ${H[i].write((time, Scan(MEM)))}$\\

\line{D6} \> \> {\bf end if}\\

\line{D7} \> \> $res \leftarrow MEM[k].op(args)$\\

\line{D8} \> \> ${T[i].write(t+2)}$ \\

\line{D9} \> \> {\bf return} $res$ \\

{\bf End} $Update$\\ \\

{\bf Operation} ${Scan()}$:\\

\line{D10} \> \> $t_0 \leftarrow {TIME.read()}$\\

\line{D11} \> \> $HELP.fetch\&add(1)$\\

\line{D12} \> \> {\bf while $\sf true$ do} \\

\line{D13} \> \> \> {\bf if} $\exists j, \hbox{timestamp in } {H[j].read()} > t_0$ {\bf then}\\
 
\line{D14} \> \> \> \>  $s \leftarrow$ snapshot in ${H[j].read()}$ \\

\line{D15} \> \> \> \>   {\bf go to} line~\ref{D32} \\

\line{D16} \> \> \> {\bf end if}\\

\line{D17} \> \> \> $t \leftarrow Collect(T)$\\

\line{D18} \> \> \> $s \leftarrow Collect(MEM)$\\

\line{D19} \> \> \> $t' \leftarrow Collect(T)$\\

\line{D20} \> \> \> {\bf if} $\forall j, t'[j] - t[j] \leq 1$ {\bf then}\\

\line{D21} \> \> \> \> $\ell \leftarrow n - |\{j : t[j] = t'[j] \, \wedge \, t[j] \hbox{is even}\}|$\\

\line{D22} \> \> \> \> {\bf if} $\ell \leq 1$ {\bf then go to} line~\ref{D32}\\

\line{D23} \> \> \> \> {\bf else}\\

\line{D24} \> \> \> \> \> {\bf for} $\lfloor \ell/2 \rfloor$ {\bf times do} \\

\line{D25} \> \> \> \> \> \> {\bf if} $s \neq Collect(MEM)$ {\bf then go to} line~\ref{D12} {\bf end if}\\

\line{D26} \> \> \> \> \> {\bf end for}  \\

\line{D27} \> \> \> \> \> $t'' \leftarrow Collect(T)$\\

\line{D28} \> \> \> \> \> {\bf if} $t' = t''$ {\bf then go to} line~\ref{D32} {\bf end if}\\

\line{D29} \> \> \> \> {\bf end ifelse}\\

\line{D30} \> \> \> {\bf end if}\\

\line{D31} \> \> {\bf end while} \\

\line{D32} \> \> $HELP.fetch\&add(-1)$\\

\line{D33} \> \> {\bf return} $s$\\

{\bf End} $Scan$\\ \\

{\bf Operation} {$Read(k)$}:\\

\line{D34} \> \> {\bf return} $MEM[k].read()$ \\

{\bf End} $Read$

\end{tabbing}
\end{minipage}
  }
\caption{A concurrent-updater algorithm that improves the algorithm~in~Figure~\ref{fig:concurrent-updater}. 
{Algorithm} for process~$p_i$. Local variables are denoted in lowercase letters.
$Update$ has constant step complexity in executions with no concurrent $Scan$ operations. 
}
\label{fig:concurrent-updater-improved}
}
\end{figure}%algorithm}
%=================================================================

\end{document}